\documentclass[12pt,a4paper]{article}

\usepackage{geometry}

\usepackage{amsmath,amsthm}
\usepackage[utf8]{inputenc}

\usepackage[bitstream-charter]{mathdesign}
\usepackage{microtype}

\usepackage{xcolor}
\definecolor{darkgreen}{rgb}{0,0.6,0}

\usepackage[bookmarks=true,bookmarksnumbered=true,pdfpagemode=UseOutlines,pdfstartview=FitH,
    pdfpagelayout=OneColumn,colorlinks=true,urlcolor=blue,citecolor=darkgreen,pdfborder={0 0 0},colorlinks]{hyperref}

\usepackage{url}
\usepackage[numbers]{natbib}

\theoremstyle{plain}
\newtheorem{theorem}{Theorem}[section]
\newtheorem{lemma}[theorem]{Lemma}
\newtheorem{proposition}[theorem]{Proposition}
\newtheorem{corollary}[theorem]{Corollary}

\theoremstyle{definition}
\newtheorem{definition}[theorem]{Definition}
\newtheorem{remark}[theorem]{Remark}

\DeclareMathOperator{\val}{val}
\DeclareMathOperator\wronskian{wr}
\newcommand\Fcal{\mathcal F}
\newcommand\Qcal{\mathcal Q}
\newcommand\FFp{\mathbb F_{\!p}}
\newcommand\FFps{\mathbb F_{\!p^s}}
\newcommand\KK{\mathbb K}
\newcommand\QQ{\mathbb Q}
\newcommand\QQp{\mathbb Q_p}
\newcommand\CC{\mathbb C}
\newcommand\Qbar{\overline\QQ}
\newcommand\ZZ{\mathbb Z}

\newcommand\NP{\mathsf{NP}}
\newcommand\mult[2]{\mu_{#1}(#2)}

\newcommand\Theorem[1]{Theorem~\ref{thm:#1}}
\newcommand\Definition[1]{Definition~\ref{def:#1}}
\newcommand\Lemma[1]{Lemma~\ref{lemma:#1}}
\newcommand\Proposition[1]{Proposition~\ref{prop:#1}}
\newcommand\Equation[1]{Equation~\eqref{eq:#1}}
\newcommand\Section[1]{Section~\ref{sec:#1}}

\newcommand\doi[1]{\textsc{doi}:\href{http://doi.org/#1}{\nolinkurl{#1}}}

\title{Computing the multilinear factors of lacunary polynomials without heights}

\author{
     Arkadev Chattopadhyay\thanks{School of Technology and Computer Science, Tata Institute of Fundamental Research, Mumbai, India; \texttt{arkadev.c@tifr.res.in}.}
\and Bruno Grenet\thanks{LIRMM, Univ. Montpellier, CNRS, Montpellier, France; \texttt{bruno.grenet@lirmm.fr}.}
\and Pascal Koiran\thanks{Univ Lyon, EnsL, UCBL, CNRS, LIP, F-69342, LYON Cedex 07, France; \texttt{[pascal.koiran,natacha.portier]@ens-lyon.fr}.}
\and Natacha Portier\footnotemark[3]
\and Yann Strozecki\thanks{DAVID laboratory, Université de Versailles Saint-Quentin, France; \texttt{yann.strozecki@uvsq.fr}.}
}

\date{}

\begin{document}
\maketitle

\begin{abstract}
We present a deterministic algorithm which computes the multilinear factors of multivariate lacunary polynomials over number fields. Its complexity is polynomial in $\ell^n$ where $\ell$ is the lacunary size of the input polynomial and $n$ its number of variables, that is in particular polynomial in the logarithm of its degree. We also provide a randomized algorithm for the same problem of complexity polynomial in $\ell$ and $n$.

Over other fields of characteristic zero and finite fields of large characteristic, our algorithms compute the multilinear factors having at least three monomials of multivariate polynomials. Lower bounds are provided to explain the limitations of our algorithm. As a by-product, we also design polynomial-time deterministic polynomial identity tests for families of polynomials which were not known to admit any.

Our results are based on so-called Gap Theorem which reduce high-degree factorization to repeated low-degree factorizations. While previous algorithms used Gap Theorems expressed in terms of the heights of the coefficients, our Gap Theorems only depend on the exponents of the polynomials. This makes our algorithms more elementary and general, and faster in most cases.
\end{abstract}

\section{Introduction} 

The \emph{lacunary representation} 
of a polynomial \[P(X_1,\dotsc,X_n)=\sum_{j=1}^k a_j X_1^{\alpha_{1,j}}\dotsm X_n^{\alpha_{n,j}}\]
is the list of the tuples $(a_j, \alpha_{1,j},\dotsc,\alpha_{n,j})$ for $1\le j\le k$. 
This representation, sometimes called \emph{sparse} or \emph{supersparse representation} in the literature, 
allows very high degree polynomials to be represented in a concise manner. The factorization of lacunary polynomials has been investigated in a series of papers. \citet{CuKoiSma99} 
first proved that integer roots of univariate integer lacunary polynomials can be found in polynomial time. 
This result was generalized by \citet{Len99} 
who proved that low-degree factors of univariate lacunary polynomials over algebraic number fields can also be found in polynomial time. 
More recently, \citet{KaKoi05,KaKoi06} 
generalized Lenstra's results to bivariate and then multivariate lacunary polynomials. 
A common point to these algorithms is that they all rely on a so-called \emph{Gap Theorem}: If $F$ is a factor of $P(X)=\sum_{j=1}^k a_j X^{\alpha_j}$, then there exists $\ell$ such that $F$ is a factor of both $\sum_{j=1}^{\ell} a_j X^{\alpha_j}$ and $\sum_{j=\ell+1}^k a_j X^{\alpha_j}$. (Here, $X$ is a vector of variables of length at least $1$, and the $\alpha_j$'s are vectors 
of exponents.) Moreover, the different Gap Theorems in these papers are all based on the notion of height of an algebraic number, and some of them use quite sophisticated results of number theory.

In this paper, we are interested in more elementary proofs for these results. We first focus on the case of linear factors of bivariate lacunary polynomials as  
\citet{KaKoi05}. Yet unlike their result our algorithm works over number fields, and is extended to multivariate polynomials, and to the computation of multilinear factors, with multiplicities. This was investigated by 
    \citet{KaKoi06}, who also dealt with low-degree factors. \citet{Gre14,Gre16} has generalized our approach to low-degree factors, subsuming all known previous results.

We prove a new Gap Theorem that does not depend on the height of an algebraic number but only on the exponents. In particular, our Gap Theorem is valid for any field of characteristic zero and we also extend it to the case of multilinear factors of multivariate polynomials. As a result, we get a new, more elementary algorithm for finding multilinear factors of multivariate lacunary polynomials over an algebraic number field. In particular, this new algorithm is easier to implement (see \citep{Gre15} for a recent implementation)
since there is no need to explicitly compute some constants from number theory, and the use of the Gap Theorem does not require the evaluation of 
the heights of the coefficients of the polynomial. 
We also compute the multiplicities of the factors. 
With our method this comes for free, which makes our algorithm faster than the previous ones in most cases (see Section~\ref{sec:comparison} for details).

Our algorithm can also be used for absolute factorization, that is the factorization of a polynomial in an algebraic closure of the field generated by its coefficients. More precisely, it can be used to compute in polynomial time the multilinear factors with at least three monomials of a lacunary multivariate polynomial. Note that univariate factorization reduces to the computation of binomial factors. And since the absolute factorization of a univariate polynomial of degree $d$ is a product of $d$ linear factors, these factors cannot even be listed in polynomial time. We shall also discuss the application of our algorithms to other fields of characteristic zero.

We use the same methods to prove a Gap Theorem for polynomials over some fields of positive characteristic, yielding an algorithm to find multilinear factors of multivariate lacunary polynomials 
with at least three monomials. We show that detecting the existence of binomial factors, that is factors with exactly two monomials, is $\NP$-hard.
This follows from the fact that  finding univariate linear factors over finite fields is $\NP$-hard~\citep{KiSha99,BiCheRo12,KaLe13}. 
In algebraic number fields we can find \emph{all} multilinear factors in polynomial time, even the binomial ones. For this we rely, as Kaltofen and Koiran did, 
on Lenstra's univariate algorithm~\citep{Len99}.

Our Gap Theorems are based on the \emph{valuation} of a univariate polynomial, that is the maximum integer $v$ such that $X^v$ divides the polynomial. We give an upper bound on the valuation of a nonzero polynomial
\[P(X)=\sum_{j=1}^k a_j X^{\alpha_j}(vX+t)^{\beta_j}(uX+w)^{\gamma_j}\text.\]
This bound can be viewed as an extension of a result due to 
\citet[see also \citealp{MoSch77}]{Hajos}. 
We also note that \citet{KaSa11} 
recently used the valuation of square roots of polynomials to make some progress on the ``Sum of Square Roots'' problem. 

Lacunary polynomials have been studied with respect to other computational tasks. For instance, \citet{Pla77} 
showed the $\NP$-hardness of computing the greatest common divisor (gcd) of two univariate integer lacunary polynomials. 
His results were extended to prove that testing the irreducibility of a lacunary polynomial is $\NP$-hard for polynomials with coefficients 
in $\ZZ$ or in a finite field~\citep{vzGaKaShp96,KaShp99,KaKoi05}. On the other hand, some 
efficient algorithms for lacunary polynomials have also been 
given, for instance for testing if a polynomial vanishes on roots of unity~\citep{Che07}, if a polynomial has a real root~\citep{BiRoSte09}, or 
for the detection of perfect powers~\citep{GieRo08,GieRo11} or interpolation~\citep{KaNe11}. 

Another approach for computing with lacunary polynomials is to give algorithms with a polynomial complexity in the logarithm of the degree (that is in the size of the exponents) but not in the number of terms or the size of the coefficients. 
This approach has been used to circumvent Plaisted's $\NP$-hardness result on the gcd~\citep{FiGraSch08,AmLeSo15}.

Note that for all the problems we address, there exist algorithms with a polynomial complexity in the degree of the polynomials. They are used as subroutines of our algorithms. We refer the reader to~\citep{vzGaGe13} for details and references on these algorithms. 

A preliminary version of this paper was published in the conference ISSAC 2013~\citep{ChaGreKoiPoStr13} that contains the bivariate case of our results. The present paper gives more details on the algorithms especially for the computation of the multiplicities of the factors, and the generalization to multivariate polynomials is new. In positive characteristic, it includes a more general $\NP$-hardness result. We also give a new Polynomial Identity Testing algorithm for sums of products of dense polynomials.

\subsection*{Organization of the paper}
\Section{wrkval} is devoted to our main technical results. In \Section{val}, we give a bound on the valuation of a nonzero polynomial $P=\sum_ja_jX^{\alpha_j}(uX+v)^{\beta_j}$ over a field of characteristic $0$. This result is extended in \Section{val:gen} to polynomials of the form $\sum_j a_j\prod_i f_i^{\alpha_{ij}}$ where the $f_i$'s are low-degree polynomials. In \Section{tight}, we discuss the tightness of these results. 

In \Section{gap}, we use these valuation bounds to get new Gap Theorems, respectively adapted to linear and multilinear factors. We also give in \Section{pit} a first application of these Gap Theorems: We give polynomial identity testing algorithms for the above mentioned families of polynomials.

\Section{algo} presents our main application: the factorization of lacunary polynomials. We begin with the computation of linear factors of bivariate polynomials over a number field in \Section{linear}, and then extend it to multilinear factors in \Section{multilinear}. Then, we generalize these algorithms to multivariate polynomials in \Section{nvar}. We briefly discuss absolute factorization and factorization in other fields of characteristic zero in \Section{otherFields}. \Section{comparison} presents a comparison between our techniques and existing ones.

Finally, the case of positive characteristic is investigated in \Section{posChar}. We show how to partially extend the results of \Section{wrkval} to positive characteristic, and we give similar algorithms as in \Section{algo}. We note that these algorithms are less general, but we also give $\NP$-hardness results explaining this lack of generality.

To understand the basics of our method, one can focus on the computation of linear factors of bivariate polynomials over a number field. To this end, one only has to read \Section{val}, \Theorem{gap} and its proof in \Section{gap}, and \Section{linear}.

\subsection*{Acknowledgments}
We wish to thank Sébastien Tavenas for his help on Proposition~\ref{prop:LowerBound}, 
and Erich L. Kaltofen for pointing us out a mistake in a previous version of Theorem~\ref{thm:FactorPosChar}.
We are grateful to the two anonymous reviewers for very interesting remarks and suggestions that improved the presentation of this paper.

\section{Wronskian and valuation}\label{sec:wrkval} 

In this section, we consider a field $\KK$ of characteristic zero and polynomials over $\KK$.
We denote by $\val(P)$ the valuation of a univariate polynomial $P$ over $\KK$, that is the largest $v$ such that $X^v$ divides $P$.  More generally, we denote by $\mult{F}{P}$ the \emph{multiplicity of $F$ as a factor of $P$}, that is the largest $m$ such that $F^m$ divides $P$. In particular, $\val(P) = \mult{X}{P}$ for $P\in\KK[X]$.

\subsection{Valuation upper bound}\label{sec:val} 

\begin{theorem}\label{thm:val}
Let $P=\sum_{j=1}^{\ell} a_j X^{\alpha_j} (uX+v)^{\beta_j}$ with $\alpha_1\leq \dotsb \leq \alpha_\ell$ and $uv\neq0$.
If $P$ is nonzero, its valuation is at most $\max_j (\alpha_j+\binom{\ell+1-j}{2})$.
\end{theorem}

Our proof of \Theorem{val} is based on the so-called \emph{Wronskian} of a family of polynomials. This is a classical tool for the study of differential equations but it has also been used 
in the field of algebraic complexity~\citep{GriKa93, LaSa96, KaSa11,KoiPoTa15,FoSaShp14}.

\begin{definition}
Let $f_1,\dotsc,f_\ell\in\KK[X]$. Their \emph{Wronskian} is the determinant of the \emph{Wronskian matrix}
\[\wronskian(f_1,\dotsc,f_\ell)=\det\begin{bmatrix}
    f_1             & f_2               & \dotsb & f_\ell   \\
    f_1'            & f_2'              & \dotsb & f_\ell'  \\
    \vdots          & \vdots            &        & \vdots   \\
    f_1^{(\ell-1)}  & f_2^{(\ell-1)}    & \dotsb & f_\ell^{(\ell-1)} 
\end{bmatrix}.\]
\end{definition}

The main property of the Wronskian is its relation to linear independence. 
The following result is classical (see \citep{BoDu10} for a simple proof of this fact).

\begin{proposition}\label{prop:wronskian}
The Wronskian of any polynomials $f_1,\dotsc,f_\ell\in\KK[X]$ 
is nonzero if and only if the $f_j$'s are linearly independent over $\KK$.
\end{proposition}

To prove \Theorem{val}, an easy lemma on the valuation of the Wronskian is needed.

\begin{lemma}\label{lemma:valinf}
Let $f_1,\dotsc,f_\ell\in\KK[X]$. Then
\[\val(\wronskian(f_1,\dotsc,f_\ell))\ge \sum_{j=1}^\ell \val(f_j)-\binom{\ell}{2}.\]
\end{lemma}

\begin{proof}
Each term of the determinant is a product of $\ell$ terms, one from each column and one from each row. The valuation of such a term is at least $\sum_j\val(f_j)-\sum_{i=1}^{\ell-1} i$ since for all $i$, $j$, $\val(f_j^{(i)})\ge\val(f_j)-i$. The result follows.
\end{proof}

The previous lemma is combined with a bound on the valuation of a specific Wronskian.

\begin{lemma}\label{lemma:valsup}
Let $f_j=X^{\alpha_j}(uX+v)^{\beta_j}$, $1\le j\le \ell$, such that $\alpha_j,\beta_j\ge \ell$ for all $j$ and $uv\neq0$. If the $f_j$'s are linearly independent, then
\[\val(\wronskian(f_1,\dotsc,f_\ell))\le \sum_{j=1}^\ell\alpha_j.\]
\end{lemma}

\begin{proof}
By Leibniz rule, for all $i$, $j$
\begin{equation*}\label{eq:leibniz} 
f_j^{(i)}(X)=\sum_{t=0}^i \binom{i}{t} (\alpha_j)_t(\beta_j)_{i-t}u^{i-t} X^{\alpha_j-t}(uX+v)^{\beta_j-i+t}
\end{equation*}
where $(m)_n=m(m-1)\dotsb(m-n+1)$ is the falling factorial.
Since $\alpha_j-t\ge\alpha_j-i$ and $\beta_j-i+t\ge\beta_j-i$ for all $t$, 
\begin{equation*}
f_j^{(i)}(X)=X^{\alpha_j-i}(uX+v)^{\beta_j-i}
        \times \sum_{t=0}^i\binom{i}{t}(\alpha_j)_t(\beta_j)_{i-t}u^{i-t} X^{i-t}(uX+v)^t. 
\end{equation*}
Furthermore, since $\alpha_j\ge \ell\ge i$, we can write $X^{\alpha_j-i}=X^{\alpha_j-\ell}X^{\ell-i}$ and since $\beta_j\ge \ell\ge i$, $(uX+v)^{\beta_j-i}=(uX+v)^{\beta_j-\ell}(uX+v)^{\ell-i}$. 
The entries of the Wronskian matrix $W$ of $f_1$, \dots, $f_\ell$ are $W_{ij} = f^{(i)_j}$ for $0\le i <\ell$ and $1\le j\le \ell$. Thus, the entries of the column $j$ of $W$ share the common factor $X^{\alpha_j-\ell}(uX+v)^{\beta_j-\ell}$, and the entries of the row $i$ share the common factor $X^{\ell-i}(uX+v)^{\ell-i}$.
Together, we get
\[\wronskian(f_1,\dotsc,f_\ell)=X^{\sum_j\alpha_j-\binom{\ell}{2}}(uX+v)^{\sum_j\beta_j-\binom{\ell}{2}}\det(M)\]
where the matrix $M$ is defined by
\[M_{i,j}= \sum_{t=0}^i \binom{i}{t} (\alpha_j)_t(\beta_j)_{i-t}u^{i-t} X^{i-t}(uX+v)^t.\]
The polynomial $\det(M)$ is nonzero since the $f_j$'s are supposed linearly independent and its degree is at most $\binom{\ell}{2}$. Therefore its valuation cannot be larger than its degree and is bounded by $\binom{\ell}{2}$. 

Altogether, the valuation of the Wronskian is bounded by $\sum_j\alpha_j-\binom{\ell}{2}+\binom{\ell}{2}=\sum_j\alpha_j$.
\end{proof}

\begin{proof}[Proof of Theorem~\protect\ref{thm:val}]
Let $P=\sum_j a_jX^{\alpha_j}(uX+v)^{\beta_j}$, and let $f_j=X^{\alpha_j}(uX+v)^{\beta_j}$. 
We assume first that $\alpha_j,\beta_j\ge \ell$ for all $j$, and that the $f_j$'s are linearly independent. Note that $\val(f_j)=\alpha_j$ for all $j$.

Let $W$ denote the Wronskian of the $f_j$'s. We can replace $f_1$ by $P$ in the first column of the Wronskian matrix 
using column operations which multiply the determinant by $a_1$ (its valuation does not change). The matrix we obtain is the Wronskian matrix of $P,f_2,\dotsc,f_\ell$. Now using \Lemma{valinf}, we get
\[\val(W)\ge \val(P)+\sum_{j\ge 2} \alpha_j-\binom{\ell}{2}.\]
This inequality combined with \Lemma{valsup} shows that
\begin{equation}\label{eq:valLinIndep} 
\val(P)\le \alpha_1+\binom{\ell}{2}.
\end{equation}

We now aim to remove our two previous assumptions. If the $f_j$'s are not linearly independent, we can extract from this family a basis $f_{j_1}, \dots, f_{j_d}$. Then $P$ can be expressed in this basis as $P=\sum_{l=1}^d \tilde a_l f_{j_l}$. We can apply \Equation{valLinIndep} 
to $f_{j_1}$,\dots, $f_{j_d}$ and obtain $\val(P)\le \alpha_{j_1} + \binom{d}{2}$. 
Since $j_d\le \ell$, we have $j_1+d-1\le \ell$ and $\val(P)\le\alpha_{j_1}+\binom{\ell+1-j_1}{2}$. The value of $j_1$ being unknown, we conclude that 
\begin{equation}\label{eq:valMax} 
\val(P)\le\max_{1\le j\le \ell}\left(\alpha_j+\binom{\ell+1-j}{2}\right).
\end{equation}

The second assumption is that $\alpha_j,\beta_j\ge \ell$. Given $P$, consider $\tilde P=X^\ell (uX+v)^\ell P=\sum_j a_j X^{\tilde{\alpha}_j}(uX+v)^{\tilde{\beta}_j}$. Then $\tilde P$ satisfies $\tilde\alpha_j,\tilde\beta_j\ge \ell$, whence by \Equation{valMax}, $\val(\tilde P)\le \max_j(\tilde{\alpha}_j+\binom{\ell+1-j}{2})$. Since $\val(\tilde P)=\val(P)+\ell$ and $\tilde{\alpha}_j=\alpha_j+\ell$, the result follows.
\end{proof}

\subsection{Generalization} \label{sec:val:gen} 

We first state a generalization of \Theorem{val} to a sum of product of powers of low-degree polynomials. Then we state a special case of this generalization that is useful for computing multilinear factors.

\begin{theorem}\label{thm:val:gen}
Let 
$(\alpha_{i,j})\in\ZZ_+^{m\times \ell}$ and
\[P=\sum_{j=1}^\ell a_j \prod_{i=1}^m f_i^{\alpha_{i,j}}\in\KK[X],\]
where the degree of $f_i\in\KK[X]$ is $d_i$ for all $i$. 
Let $F\in\KK[X]$ be an irreducible polynomial and let $\mu_i=\mult{F}{f_i}$ for all $i$. Then the multiplicity $\mult{F}{P}$ of $F$ as a factor of $P$ satisfies
\[\mult{F}{P}\le \max_{1\le j\le \ell}\ \sum_{i=1}^m\left( \mu_i\alpha_{i,j}+(d_i-\mu_i)\binom{\ell+1-j}{2}\right).\]
\end{theorem}

\begin{proof}
Let $P_j=\prod_{i=1}^m f_i^{\alpha_{i,j}}$ for $1\le j\le \ell$. As in the proof of \Theorem{val}, we can assume without loss of generality that the $P_j$'s are linearly independent, and the $\alpha_{i,j}$'s not less than $\ell$. 

We can use a generalized Leibniz rule to compute the derivatives of the $P_j$'s. Namely
\begin{equation}\label{eq:leibnizP}
P_j^{(T)}=\sum_{t_1+\dotsb+t_m=T} \binom{T}{t_1,\dotsc,t_m} \prod_{i=1}^m (f_i^{\alpha_{i,j}})^{(t_i)},
\end{equation}
where $\binom{T}{t_1,\dotsc,t_m}=T!/(t_1!\dotsm t_m!)$ is the multinomial coefficient. 
Consider now a derivative of the form $(f^\alpha)^{(t)}$. This is a sum of terms, each of which contains a factor $f^{\alpha-t}$. 
In \Equation{leibnizP}, each $t_i$ is bounded by $T$. This means that $P_j^{(T)}=Q_{T,j} \prod_i f_i^{\alpha_{i,j}-T}$ for some polynomial $Q_{T,j}$. 
Since the degree of $P_j^{(T)}$ equals $\sum_i d_i\alpha_{i,j}-T$, 
$Q_{T,j}$ 
has degree $\sum_id_i\alpha_{i,j}-T-\sum_i(d_i\alpha_{i,j}-d_iT)=(\sum_id_i-1)T$. 

Consider now $W=\wronskian(P_1,\dotsc,P_\ell)$. We can factor out  $\prod_i f_i^{\alpha_{i,j}-\ell}$
in each column and $\prod_if_i^{\ell-T}$ in each row. 
At row $T$ and column $j$, we therefore factor out $\prod_i f_i^{\alpha_{i,j}-\ell} \cdot\prod_i f_i^{\ell-T}=\prod_i f_i^{\alpha_{i,j}-T}$.
Thus, 
\[W = \prod_{i=1}^m f_i^{\sum_j\alpha_{i,j}-\binom{\ell}{2}} \det(M)\]
where $M_{T,j}=Q_{T,j}$. Thus, $\det(M)$ is a polynomial of degree at most $(\sum_id_i-1)\binom{\ell}{2}$. 

Therefore, the multiplicity $\mult{F}{W}$ of $F$ as a factor of $W$ is bounded by its multiplicity as a factor of $\prod_i f_i^{\sum_j\alpha_{i,j}-\binom{\ell}{2}}$ plus the degree of $\det(M)$. We get
\begin{align} 
\mult{F}{W}
    &\le \sum_i \mu_i\left(\sum_j\alpha_{i,j}-\binom{\ell}{2}\right)+(\sum_id_i-1)\binom{\ell}{2} \nonumber\\
    &=\sum_i\left(\mu_i\sum_j\alpha_{i,j}+(d_i-\mu_i)\binom{\ell}{2}\right)-\binom{\ell}{2}.\label{eq:multUpBd}
\end{align}

To conclude the proof, it remains to remember \Lemma{valinf} and use the same proof technique as in \Theorem{val}. It was expressed in terms of the valuation of the polynomials, but remains valid with the multiplicity of any factor. In this case, it can be written as $\mult{F}{W}\ge\sum_j\mult{F}{P_j} -\binom{\ell}{2}$ where $W$ is the Wronskian of the $P_j$'s. Using column operations, we can replace the first column of the Wronskian matrix of the $P_j$'s by the polynomial $P$ and its derivatives. We get $\mult{F}{W}\ge\mult{F}{P}+\sum_{j\ge 2}\mult{F}{P_j}-\binom{\ell}{2}$, where $\mult{F}{P_j}=\sum_i\mu_i\alpha_{i,j}$.

Together with \eqref{eq:multUpBd}, 
we get
\begin{align*}
\mult{F}{P}& \le \mult{F}{W}-\sum_{j\ge 2}\mult{F}{P_j}+\binom{\ell}{2} \\
           & \le \sum_i\left(\mu_i\sum_j\alpha_{i,j}+(d_i-\mu_i)\binom{\ell}{2}\right)-\binom{\ell}{2}
                                                -\sum_{j\ge 2}\sum_i\mu_i\alpha_{i,j} 
                                                +\binom{\ell}{2}\\
           & \le \sum_i\left(\mu_i\alpha_{i1} +(d_i-\mu_i)\binom{\ell}{2}\right).
\end{align*}

To obtain the bound of the theorem, the two initial assumption have to be removed using the same technique as in \Theorem{val}.
\end{proof}

In the previous proof, the multiplicity of $F$ as a factor of $\det(M)$ is bounded by the degree of $\det(M)$. One can actually bound this multiplicity by the degree of $\det(M)$ divided by the degree of $F$. If we denote by $d_F$ the degree of $F$, the bound of the theorem can be refined to
\[\mu_F(P)\le\max_{1\le j\le\ell} \left[\sum_{i=1}^m\left(\mu_i\alpha_{i,j} + \left(\frac{d_i}{d_F} -\mu_i\right)\binom{\ell+1-j}{2}\right) + \left(1-\frac{1}{d_F}\right)\binom{\ell+1-j}{2}\right].\]

As a special case of the theorem, 
one obtains the following corollary.

\begin{corollary} \label{cor:val:3terms}
Let $P=\sum_{j=1}^\ell a_j X^{\alpha_j}(vX+t)^{\beta_j}(uX+w)^{\gamma_j}$, $wt\neq0$. If $P$ is nonzero then 
its valuation is at most
$\max_{1\le j\le \ell} (\alpha_j+2\binom{\ell+1-j}{2})$.
\end{corollary}

\subsection{Is \Theorem{val} tight?} \label{sec:tight}

Let $P$ be as in \Theorem{val}, that is 
\[P=\sum_{j=1}^\ell a_j X^{\alpha_j}(uX+v)^{\beta_j}\]
where $uv\neq0$. 
If the family $(X^{\alpha_j}(1+X)^{\beta_j})_{1\le j\le \ell}$ is linearly independent over $\KK$, the valuation of $P$ is at most $\alpha_1+\binom{\ell}{2}$. 
In this section, we investigate the tightness of this bound.
In the special case $\alpha_j=\alpha_1$ for all $j$, Haj\'os' Lemma~\citep{Hajos, MoSch77}
gives the better bound $\alpha_1+(\ell-1)$. (This bound can be shown to be tight by expanding ${X^{\ell-1}=(-1+(X+1))^{\ell-1}}$ with the binomial formula.)

We first prove that Haj\'os bound does not apply when the $\alpha_j$'s are not all equal by constructing a family of polynomials of valuation $2\ell-3$ where $\ell$ is the number of terms. Note that in the following proposition, the polynomial has $\ell'=\ell+3$ terms and valuation $2\ell'-3 = 2\ell+3$.

\begin{proposition}\label{prop:LowerBound}
For $\ell\ge 0$,
\[X^{2\ell+3} = -1 + (1+X)^{2\ell+3} + \sum_{j=0}^\ell \frac{2\ell+3}{2j+1}\binom{\ell+1+j}{\ell+1-j} X^{2j+1}(1+X)^{\ell+1-j}.\]
Furthermore, the family $(1, (1+X)^{2\ell+3}, X(1+X)^{\ell+1}, X^3(1+X)^\ell,\dots,X^{2\ell+1}(1+X))$ is linearly independent.
\end{proposition}

\begin{proof}
It is clear that the polynomial 
has degree $(2\ell+3)$ and is monic. Furthermore, the family is linearly independent since the polynomials have pairwise different degrees.

Let $P_\ell$ be the right-hand side polynomial in the proposition and $[X^m]P_\ell$ be the coefficient of the monomial $X^m$ in $P_\ell$. We shall prove that $[X^m]P_\ell=0$ for all $m<2\ell+3$. It is clear for $m=0$. For $m>0$,
\begin{equation}\label{eq:coeff}
[X^m]P_\ell = \binom{2\ell+3}{m} -                                   
              \sum_{j=0}^\ell \frac{2\ell+3}{2j+1}\binom{\ell+1+j}{\ell+1-j}\binom{\ell+1-j}{m-2j-1}.
\end{equation}

We prove that $[X^m]P_\ell=0$ with the help of the implementation of 
Wilf and Zeilberger's algorithm 
in the Maple package \texttt{EKHAD} of Doron Zeilberger 
\citep{PetkovsekWilfZeilberger}. 
For $0\le j\le\ell$, let $F(m,j)$ be the summand in equation~\eqref{eq:coeff} 
divided by $\binom{2\ell+3}{m}$. 
The package \texttt{EKHAD} provides the recurrence 
\[ mF(m+1,j)-mF(m,j)
=F(m,j+1)R(m,j+1)-F(m,j)R(m,j)\]
where
\[R(m,j)=\frac{2j(2j+1)(\ell+j+2-m)}{(2\ell+3-m)(2j-m)}.\]
To conclude the proof, it is sufficient to prove the recurrence relation and that it implies the equality. These are simple computations that we omit but can easily be checked by hand. That is, the validity of the proof does not rely in any manner on the validity of the computer program that was used to produce it.
\end{proof}

We now address the question of improving the bound in \Theorem{val}. The bound is obtained as a combination of a lower bound, given by \Lemma{valinf}, and an upper bound, given by \Lemma{valsup}. 
None of the bounds can be improved in the general case: The lower bound is tight if all the valuations are distinct while the upper bound is tight if all the valuations are equal. On the other hand, both bounds can be improved in some special cases. Let $f_1$, \dots, $f_\ell\in\KK[X]$ ordered by increasing valuations, such that $\val(f_j)\le\val(f_1)+(j-1)$ for all $j>0$. Then $\val(\wronskian(f_1,\dots,f_\ell))\ge\ell\val(f_1)$ since for all $j$, $f_j$ can be replaced in the Wronskian matrix by some $g_j$ of valuation at least $\val(f_1)+(j-1)$ (proof omitted). And if on the contrary the valuations are all distinct, the matrix made of the constant coefficients of the entries of the matrix $M$ defined in the proof of \Lemma{valsup} is similar to a Vandermonde matrix. This implies that the valuation of the Wronskian equals $\sum_j \alpha_j-\binom\ell{2}$. These two improvements are enough to recover Haj\'os' bound on the one hand, and the obvious fact that $\val(P)=\alpha_1$ if the $\alpha_j$'s are pairwise distinct on the other hand. Unfortunately, they are not sufficient to improve \Theorem{val} in the general case. For instance, if $\val(f_j) = \val(f_1)+\binom{j-1}{2}$ for all $j$, we are not able to improve any of the two bounds.

\section{Gap Theorems and their application to PIT}\label{sec:gap} 

In this section, we first prove our Gap Theorems. Then we give very direct applications of these theorems in the form of Polynomial Identity Testing algorithms for some families of univariate polynomials.

\subsection{Two Gap Theorems} 

We still assume that the coefficients of the polynomials we consider lie in some field $\KK$ of characteristic zero.

The bound on the valuation obtained in the \Section{val} translates into a Gap Theorem for linear factors.

\begin{theorem}[Gap Theorem for linear factors] \label{thm:gap}
Let $P=Q+R$ where
\[Q=\sum_{j=1}^\ell a_j X^{\alpha_j}Y^{\beta_j}\text{ and }R=\sum_{j=\ell+1}^k a_j X^{\alpha_j}Y^{\beta_j}\]
such that $\alpha_1\le\dotsb\le\alpha_k$. Suppose that 
$\ell$ is the smallest index such that $\alpha_{\ell+1}>\alpha_1+\binom{\ell}{2}$.
Then, for every $F=uX+vY+w$ with $uvw\neq0$, 
\[\mult{F}{P}=\min(\mult{F}{Q},\mult{F}{R})\text.\]
\end{theorem}

\begin{proof}
Let $F=uX+vY+w$, $uvw\neq0$. Then $F$ divides $P$ if and only if $P(X, -\frac{1}{v}(uX+w))=0$, and the same holds for the polynomials $Q$ and $R$. Let $P^\star(X)= P(X, -\frac{1}{v}(uX+w))$, and define $Q^\star$ and $R^\star$ in the same way from $Q$ and $R$.

Let us first prove that $F$ divides $P$ if and only if it divides both $Q$ and $R$. 
Clearly, $F$ divides $P$ if it divides both $Q$ and $R$. 
Suppose 
that $F$ does not divide $Q$, that is $Q^\star$ is nonzero. By \Theorem{val}, its valuation is at most $\max_{j\le\ell}(\alpha_j+\binom{\ell+1-j}{2})$. 
Furthermore, $\alpha_{j+1}\le\alpha_1+\binom{j}{2}$ for all $j<\ell$ by hypothesis. Therefore,
\begin{align*}
\val(Q^\star)   &\le\max_{1\le j\le\ell}\left(\alpha_j+\binom{\ell+1-j}{2}\right)\\
                &\le\max_{1\le j\le\ell}\left(\alpha_1+\binom{j-1}{2}+\binom{\ell+1-j}{2}\right)\\
                &\le\alpha_1+\binom{\ell}{2}\text.
\end{align*}
The last inequality holds since $\binom{j-1}{2}+\binom{\ell-(j-1)}{2}\le\binom{\ell}{2}$ for $1\le j\le\ell$.

The valuation of $R^\star$ is at least $\alpha_{\ell+1}$ which is by hypothesis larger than $\alpha_1+\binom{\ell}{2}$.
Therefore, if $Q^\star$ is not identically zero, its monomial of lowest degree cannot be canceled by a monomial of $R^\star$. In other words, $P^\star=Q^\star+R^\star$ is nonzero and $F$ does not divide $P$.
Finally, if $F$ divides $Q$ but not $R$, it cannot divide $P$ since it would divide $P-Q = R$. 

To show that $\mult{F}{P}=\min(\mult{F}{Q},\mult{F}{P})$, we remark that $F$ is a factor of multiplicity $\mu$ of $P$ if and only if it divides $\partial^m P/\partial X^m$ for all $m\le\mu$. Since $\mu=\mult{F}{P}=\mult{F}{X^d P}$ for all $d$, one can assume that $\alpha_1\ge\mu$. Then
\[\frac{\partial^m P}{\partial X^m}=\sum_{j=1}^k a_j (\alpha_j)_m X^{\alpha_j-m}Y^{\beta_j}\]
for $1\le m\le\mu$. 
The hypothesis $\alpha_{\ell+1}>\alpha_1+\binom{\ell}{2}$ in the theorem only depends on the difference between the exponents. By linearity of the derivative, the previous argument actually shows that $F$ divides $\partial^m P/\partial X^m$ if and only if it divides both $\partial^m Q/\partial X^m$ and $\partial^m R/\partial X^m$. This proves that $\mult{F}{P}=\min(\mult{F}{Q},\mult{F}{R})$. 
\end{proof}

It is straightforward to extend this theorem to more \emph{gaps}. The theorem can be recursively applied to $Q$ and $R$ (as defined in the proof). Then, if $P=P_1+\dotsc+P_s$ where there is a \emph{gap} between $P_t$ and $P_{t+1}$ for $1\le t<s$, then any linear polynomial $(uX+vY+w)$ is a factor of multiplicity $\mu$ of $P$ if and only if it is a factor of multiplicity at least $\mu$ of each $P_t$. Moreover, one can write every $P_t$ as $X^{q_t} Q_t$ such that $X$ does not divide $Q_t$. Then the degree in $X$ 
of $Q_t$ is bounded by $\binom{\ell_t-1}{2}$ where $\ell_t$ is its number of terms.

The following definition makes this discussion formal.

\begin{definition}\label{def:decomposition}
Let $P$ as in \Theorem{gap}. A set $\Qcal=\{Q_1,\dotsc,Q_s\}$ is a \emph{decomposition of $P$ with respect to $X$} if there exist integers $q_1$, \dots, $q_s$ such that $P=X^{q_1} Q_1+\dotsb+X^{q_s} Q_s$ with $q_t>q_{t-1}+\deg_X(Q_{t-1})$ for $1<t\le s$.

A decomposition is \emph{compatible} with a set $\Fcal$ of polynomials if for all $F\in\Fcal$, $\mult{F}{P}=\min_{1\le t\le s} \mult{F}{Q_t}$.

The \emph{degree} in $X$ (resp. in $Y$) of a decomposition is the sum of the degrees in $X$ (resp. in $Y$) of the $Q_t$'s.
\end{definition}

The Gap Theorem implies a decomposition of $P$ of degree at most $\binom{k-1}{2}$ in $X$. Indeed, the degree of each $Q_t$ is at most $\binom{\ell_t-1}{2}$, with $\sum_t\ell_t=k$, and the function $k\mapsto\binom{k}{2}$ is super-additive, \emph{i.e.}, $\binom{i+j}{2} \ge \binom{i}{2} + \binom{j}{2}$. 

\begin{remark}\label{rk:subdecomposition}
Let $\Qcal'$ be the decomposition obtained from $\Qcal$ by replacing $Q_t$ and $Q_{t+1}$ by $Q'_t=Q_t+X^{q_{t+1}-q_t} Q_{t+1}$. It is easy to see that if $\Qcal$ is compatible with a set $\Fcal$, then so is 
$\Qcal'$. More generally, one obtains compatible decompositions by grouping together any number of consecutive polynomials in a decomposition.
\end{remark}

Using the generalization of \Theorem{val} given in \Section{val:gen}, one can also prove a similar Gap Theorem, but for multilinear factors.

\begin{theorem}[Gap Theorem for multilinear factors] \label{thm:gap:multilin}
Let $P=Q+R$ where 
\[Q=\sum_{j=1}^\ell a_j X^{\alpha_j} Y^{\beta_j}\text{ and }R=\sum_{j=\ell+1}^k a_j X^{\alpha_j} Y^{\beta_j},\]
such that $\alpha_1\leq \dots \leq \alpha_k$. If $\ell$ is the smallest index such that $\alpha_{\ell+1} > \alpha_1+2 \binom{\ell}{2}$
then for every $F=uXY+vX+wY+t$ with $wt\neq0$, 
\[\mult{F}{P}=\min(\mult{F}{Q},\mult{F}{R})\text.\]
\end{theorem}

\begin{proof}
Let $F=uXY+vX+wY+t$, $wt\neq0$. Then $F$ divides $P$ if and only if $P(X,-\frac{vX+t}{uX+w})=0$. This is equivalent to the fact that the polynomial
\[(uX+w)^B P(X,-\frac{vX+t}{uX+w})=\sum_{j=1}^k a_j X^{\alpha_j} (-vX-t)^{\beta_j} (uX+w)^{B-\beta_j}\]
vanishes, where $B=\max_j \beta_j$. The rest of the proof is identical to the proof of the Gap Theorem for linear factors, using the valuation bound of Corollary~\ref{cor:val:3terms} instead of \Theorem{val}.
\end{proof}

\subsection{Polynomial Identity Testing} \label{sec:pit} 

The first algorithmic applications of our technical results are 
two polynomial identity testing algorithms. 
The problem is the following: Given a polynomial in some specified representation, decide whether it is identically zero. Of course, if the polynomial is given as the list of its coefficient, the problem is trivial. In the most general forms of the problem, the polynomial is represented either as an arithmetic circuit or a blackbox. Then, it is not clear how to perform the test since expanding the polynomial as a sum of monomial may require exponential time. We refer the reader to the survey of \citet{ShpYe10} for more details and recent results on this problem.

Our algorithms deal with univariate high-degree polynomials, represented as sums of product of powers of dense polynomials. In particular, one can represent such polynomials by small arithmetic circuits. If the polynomial is known to be sparse, one can apply the algorithm of~\citet{BlaHaLiVi09}. \citet{KoiPoTa15} and \citet{GreKoiPoStr11} give results close to the ones we prove here. A comparison is given at the end of this section.

The polynomials in this section have coefficients in an algebraic number field $\KK=\QQ[\xi]/\langle\varphi\rangle$ where $\varphi\in\QQ[\xi]$ is irreducible. An element $e$ of $\KK$ is uniquely represented by a polynomial $p_e\in\QQ[\xi]$ of degree smaller than $\deg(\varphi)$. In the algorithms, a coefficient $c\in\KK$ of a lacunary polynomial is given as the dense representation of $p_c$, that is the list of all its coefficients including the zero ones. Moreover, the algorithms are uniform in $\KK$ in the sense that they can take as input the polynomial $\varphi$ defining $\KK$.

\begin{theorem}\label{thm:pit}
Let $\KK$ be an algebraic number field and 
\[P=\sum_{j=1}^k a_j X^{\alpha_j}(uX+v)^{\beta_j}\in\KK[X]\text.\]
There exists a deterministic polynomial-time algorithm to decide if $P$ vanishes.
\end{theorem}

\begin{proof}
We assume without loss of generality that $\alpha_{j+1}\ge\alpha_j$ for all $j$ and $\alpha_1=0$. 
If $\alpha_1$ is nonzero, $X^{\alpha_1}$ divides $P$ and we consider $P/X^{\alpha_1}$.

Suppose first that $u=0$. Then $P$ is given as a sum of monomials, and we only have to test each coefficient for zero. Note that the $\alpha_j$'s are not necessarily 
distinct. Thus the coefficients are of the form $\sum_j a_jv^{\beta_j}$. \citet{Len99} 
gives an algorithm to find low-degree factors of 
univariate lacunary 
polynomials. It is easy to deduce from his algorithm an algorithm to test such sums for zero. A strategy could be to simply apply Lenstra's algorithm to $\sum_ja_j X^{\beta_j}$ and then check whether $(X-v)$ is a factor, but one can actually improve the complexity by extracting from his algorithm the relevant part (we omit the details). 
The case $v=0$ is similar. 

We assume now that $uv\neq 0$. Then $P=0$ if and only if $(Y-uX-v)$ divides $\sum_j a_j X^{\alpha_j}Y^{\beta_j}$. Recursively using the Gap Theorem for linear factors (\Theorem{gap}), one computes a decomposition $P=X^{q_1}Q_1+\dotsb+X^{q_s}Q_s$ such that $P$ is identically zero if and only if each $Q_t$ in this sum is 
also. 
Therefore, we are left with testing if each $Q_t$ is identically zero. 

To this end, let $Q$ be one these polynomials. With a slight abuse of notation, it can be written $Q=\sum_{j=1}^k a_jX^{\alpha_j}(uX+v)^{\beta_j}$. It satisfies $\alpha_1=0$ and $\alpha_{j+1}\le\binom{j}{2}$ for all $j$. In particular, $\alpha_k\le\binom{k-1}{2}$.
Consider the change of variables $Y=uX+v$. Then
\[Q(Y)=\sum_{j=1}^k a_ju^{-\alpha_j} (Y-v)^{\alpha_j}Y^{\beta_j}\]
is identically zero if and only if $Q(X)$ is. 
We can express $Q(Y)$ as a sum of powers of $Y$: 
\[Q(Y)=\sum_{j=1}^k\sum_{\ell=0}^{\alpha_j} a_j u^{-\alpha_j} \binom{\alpha_j}{\ell} (-v)^\ell Y^{\alpha_j+\beta_j-l}.\]
There are at most $k\binom{k-1}{2}=\mathcal O(k^3)$ 
monomials. Then, testing if $Q(Y)$ is identically zero consists in testing whether each coefficient vanishes. 
Moreover, each coefficient has the form $\sum_j \binom{\alpha_j}{\ell_j} a_j u^{-\alpha_j} (-v)^{\ell_j}$ where the sum ranges over at most $k$ indices. Since $\ell_j,\alpha_j\le\binom{k-1}{2}$ for all $j$, the terms in these sums have polynomial bit-lengths. 
Therefore, the coefficients can be tested for zero in polynomial time.

Altogether, this gives a polynomial-time algorithm to test 
if $P$ is identically zero.
\end{proof}

One can actually replace the linear polynomial $(uX+v)$ in the previous theorem by any binomial polynomial. Without loss of generality, one can consider that this binomial is $(uX^d+v)$, and we assume that it is represented in lacunary representation. In other words, its size is polynomial in $\log(d)$.

\begin{corollary} 
Let $\KK$ be an algebraic number field and 
\[P=\sum_{j=1}^k a_j X^{\alpha_j}(uX^d+v)^{\beta_j}\in\KK[X]\text.\]
There exists a deterministic polynomial-time algorithm to decide if the polynomial $P$ vanishes.
\end{corollary} 

\begin{proof}
For all $j$ we consider the Euclidean division of $\alpha_j$ by $d$: $\alpha_j=q_jd+r_j$ with $r_j<d$. We rewrite $P$ as
\[P=\sum_{j=1}^k a_j X^{r_j} (X^d)^{q_j} (uX^d+v)^{\beta_j}.\]
Let us group in the sum all the terms with a common $r_j$. That is, let 
\[P_i(Y)=\sum_{\substack{1\le j\le k\\ r_j=i}} a_j Y^{q_j} (uY+v)^{\beta_j}\]
for $0\le i<d$. We remark that regardless of the value of $d$, the number of nonzero $P_i$'s is bounded by $k$. We have $P(X)=\sum_{i=0}^{d-1} X^i P_i(X^d)$. Each monomial $X^\alpha$ of $X^iP_i(X^d)$ satisfies $\alpha\equiv i\mod d$. Therefore, $P$ is identically zero if and only if all the $P_i$'s are identically zero.

Since each $P_i$ has the same form as in \Theorem{pit}, and there are at most $k$ of them, the previous algorithm can be applied to each of them to decide the nullity of $P$.
\end{proof}

Using \Theorem{val:gen} instead of \Theorem{val}, one can give a polynomial identity testing algorithm for a larger class of polynomials. 
In the next algorithm, we use the \emph{dense representation} for univariate polynomials, that is the list of all their coefficients (including the zero ones). Then \emph{dense size} of a univariate polynomial is then the sum of the sizes of its coefficients.

\begin{theorem}\label{thm:pit:gen}
Let $f_1$, \dots, $f_m$ be monic univariate polynomials over a number field. Let 
\[P=\sum_{j=1}^k a_j \prod_{i=1}^m f_i^{\alpha_{i,j}}.\]
There is a deterministic algorithm to decide if $P$ is zero whose running time is polynomial in $k$, $m$, the dense sizes of the $f_i$'s and the bitsizes of the $a_j$'s and the $\alpha_{i,j}$'s.

With an oracle to decide the nullity of sums of the form $\sum_j\prod_i \lambda_i^{\alpha_{i,j}}$ where the $\lambda_i$'s are in the number field, the above algorithm can be used with non monic polynomials $f_1$, \dots, $f_m$.
\end{theorem}

\begin{proof}
Let us suppose that the $f_i$'s are not monic. In time polynomial in the degree of the $f_i$'s, one can compute their monic irreducible factorizations. Then we can write $f_i=\lambda_i\prod_{t=1}^n g_t^{\beta_{i,t}}$ for all $i$, where the $g_t$'s are distinct monic irreducible polynomials and $\beta_{i,t}\ge 0$. Then, the polynomial $P$ can be written as
\[P=\sum_{j=1}^k \left[ a_j \biggl(\prod_{i=1}^m \lambda_i^{\alpha_{i,j}}\biggr) \biggl(\prod_{t=1}^n g_t^{\sum_i\beta_{i,t}\alpha_{i,j}}\biggr)\right]\text.\]
For all $j$ and $t$, let $\gamma_{j,t}=\sum_i\beta_{i,t}\alpha_{i,j}$ and $\Lambda_j=\prod_i\lambda_i^{\alpha_{i,j}}$. Note that the case of monic polynomials is the case where $\Lambda_j=1$.

If $n=1$, that is all the $f_i$'s are powers of a same polynomial $g$, then $P=\sum_j a_j \Lambda_j g^{\gamma_j}$. It is thus sufficient to find the subsets of indices for which $\gamma_j$ is constant, and test for zero sums of the form $\sum_j a_j \Lambda_j$. These sums can be easily tested for zero if $\Lambda_j=1$, and using the oracle otherwise.

If $n>1$, we use a Gap Theorem. To this end, we use the bound on the multiplicity of the irreducible polynomial $g_1$ given by \Theorem{val:gen}. 
Let
\[ Q=\sum_{j=1}^\ell a_j \Lambda_j \prod_{t=1}^n g_t^{\gamma_{j,t}}\text{ and } R=\sum_{j=\ell+1}^k a_j \Lambda_j \prod_{t=1}^n g_t^{\gamma_{j,t}}\]
and suppose that $\ell$ is the smallest index such that 
\begin{equation}\label{eq:pit-gap}
\gamma_{\ell+1,1}> \gamma_{1,1}+\biggl(\sum_{t>1} \deg(g_t)\biggr)\binom{\ell}{2}\text.
\end{equation}
Since $\mult{g_1}{g_t}=0$ for all $t>1$ and $\mult{g_1}{g_1}=1$, \Theorem{val:gen} implies that 
\begin{align*}
\mult{g_1}{Q}
    &\le \max_{j\le\ell}\biggl( \gamma_{1,j}+\sum_{t>1}\deg(g_t)\binom{\ell-j+1}{2}\biggr)\\
    &\le \gamma_{1,1}+\biggl(\sum_{t>1}\deg(g_t)\biggl)\binom{\ell}{2}\text.
\end{align*}
The second inequality can be proved exactly as in the proof of the Gap Theorem for linear factors. In particular, one has $\gamma_{\ell+1,1}>\mult{g_1}{Q}$. Therefore, if $Q$ is nonzero, since $\mult{g_1}{R}\ge\gamma_{\ell+1,1}$, then $P=Q+R$ is nonzero. The same argument of course works for any $g_t$.

Algorithmically, we can decompose $P$ into $Q+R$ using \Equation{pit-gap}, and recursively decompose $R$. In each polynomial of the decomposition we factor out the largest possible power of $g_1$. We obtain $P=g_1^{\gamma_{(1)}}Q_1+\dotsb+g_1^{\gamma_{(p)}}Q_p$ for some $p$ such that $P=0$ if and only if each $Q_1=\dotsb=Q_p=0$. Furthermore, the exponents of $g_1$ in $Q_1$, \dots, $Q_p$ are bounded by $\binom{k}{2}\sum_{t>1} \deg(g_t)$. 
Let us call $\{Q_1,\dotsc,Q_p\}$ the decomposition of $P$ with respect to $g_1$. 

The algorithm is then as follows. We initially consider the set $\Qcal=\{P\}$. Then for $t=1$ to $n$, we replace each member of $\Qcal$ by its decomposition with respect to $g_t$. We obtain a set $\Qcal$ of polynomials such that $P=0$ if and only if $Q=0$ for all $Q\in\Qcal$. 

It remains to test whether each polynomial of $\Qcal$ vanishes. To this end, one can expand these polynomials as sums of monomials since their degrees are polynomially bounded. The coefficients of the monomials will have the form $\sum_j a_j\Lambda_jc_{j,\delta}$ where $c_{j,\delta}$ is the coefficient of $X^\delta$ in the polynomial $\prod_t g_t^{\gamma_{j,t}}$. Therefore, it has a polynomial bitsize. Testing the nullity of these coefficients is then easy if $\Lambda_j=1$ and done using the oracle otherwise.

This proves the theorem.
\end{proof}

Note that an alternative algorithm for the same class of polynomials is given by \citet{KoiPoTa15}, 
in the case where the $f_i$'s are not monic. Yet the complexity of their algorithm is exponential in $k$ and $m$. We also remark that our algorithm uses, for non monic polynomials, an oracle to decide the nullity of a sum of products of powers of integers (or elements of a number field). The same oracle is also needed in the polynomial identity testing algorithm of \citet{GreKoiPoStr11}. This algorithm deals with a very similar family of polynomial, where the $f_i$'s are sparse polynomials instead of dense ones, but the complexity is then polynomial in the exponents $\alpha_{ij}$'s rather than in their bitsizes.

\section{Factoring lacunary polynomials} \label{sec:algo} 

We now turn to the main applications of our Gap Theorems of \Section{gap}. We detail 
how to compute linear and multilinear factors of lacunary polynomials over number fields. 
We first focus on bivariate polynomials. In \Section{nvar}, we explain that our algorithms can be easily extended to  multivariate polynomials. We discuss the case of other fields of characteristic zero in \Section{otherFields}.

The input polynomials in our algorithms are given in lacunary representation. A monomial is represented by its coefficient and its vector of exponents written in binary. 
For the representation of elements of a number field, we refer to the discussion opening \Section{pit}. The size of the lacunary representation is the sum of the sizes of the representations of the nonzero monomials. In particular, note that the size 
is polynomial in the logarithm of the degree 
when the number of terms and the coefficient heights are fixed.

As we shall see, finding binomial factors is a special case in our algorithms. To simplify the next proofs, we first prove a lemma on the computation of these factors.

\begin{lemma}\label{lemma:binomials}
Let $\KK$ be an algebraic number field and
\[P=\sum_{j=1}^k a_j X^{\alpha_j}Y^{\beta_j}\in\KK[X,Y]\text.\]
There exists a deterministic polynomial-time algorithm that computes all the multilinear binomial factors of $P$, together with their multiplicities.
\end{lemma}

\begin{proof}
Let $F$ be an irreducible binomial factor of $P$. Without loss of generality, let us assume that it depends on $X$. Otherwise, invert the roles of $X$ and $Y$ in what follows. Then $F$ can be written $F=uXY^\gamma+vY^\delta$ with $u,v\in\KK$ and $\gamma+\delta\le 1$. Then $F$ divides $P$ if and only if $P(-\frac{v}{u} Y^{\delta-\gamma},Y)=0$. In other words, $F$ divides $P$ if and only if $G=uZ+v$ is a factor of the polynomial $Q(Y,Z)=Y^{\max_j\alpha_j} P(ZY^{\delta-\gamma},Y)$. One can view $Q$ as an element of $\KK[Z][Y]$ and write it as $Q=\sum_\epsilon q_\epsilon(Z) Y^\epsilon$. Then $G$ divides $Q$ if and only if it divides each $q_\epsilon$. More precisely, the multiplicity $\mult{F}{P}$ of $F$ as a factor of $P$ equals $\min_\epsilon\mult{G}{q_\epsilon}$. Therefore, it is sufficient to compute the linear factors of each $q_\epsilon$, with multiplicities, using Lenstra's algorithm~\citep{Len99}. 

To find all the multilinear binomial factors of $P$ depending on $X$, one has to apply the above algorithm with $(\gamma,\delta)=(1,0)$, $(0,1)$ and $(0,0)$. For the factors depending only on $Y$, one has to invert the roles of $X$ and $Y$ and apply to above algorithm once more.
\end{proof}

Note that in the previous algorithm, it is not efficient 
to first compute the gcd of the polynomials $q_\epsilon$ before the computation of their common linear factors since this task is $\NP$-hard to perform~\citep{Pla77}.

\subsection{Finding linear factors} \label{sec:linear} 

\begin{theorem}\label{thm:linfact}
Let $\KK$ be a number field and 
\[P(X,Y)=\sum_{j=1}^k a_j X^{\alpha_j}Y^{\beta_j}\in\KK[X,Y].\]
Then there 
exists a deterministic polynomial-time algorithm that finds all the linear factors of $P$, together with their multiplicities.
\end{theorem}

\begin{proof}
The algorithm has three distinct parts. The first part is the obvious computation of monomial factors, the second part is for binomial factors using \Lemma{binomials}, and the third part for trinomial factors.

Consider the factors of the form $F=uX+vY+w$, $uvw\neq 0$. Using recursively the Gap Theorem for linear factors (\Theorem{gap}), we can compute a decomposition of $P$ with respect to $X$ of degree at most $\binom{k-1}{2}$, compatible with the linear trinomial factors. That is, $P$ can be written $P=X^{q_1}Q_1+\dotsb+X^{q_s}Q_s$ such that $\mult{F}{P}=\min_t \mult{F}{Q_t}$, and the sum of the degrees of the $Q_t$'s is bounded by $\binom{k-1}{2}$. 
Inverting the roles of $X$ and $Y$, one can compute a compatible decomposition of each $Q_t$ with respect to $Y$. Globally, the polynomial $P$ can be expressed as
\[P=\sum_{t=1}^s X^{\alpha_{(t)}}Y^{\beta_{(t)}} R_t\]
where each $R_t$ has $\ell_t$ terms, $\sum_t\ell_t=k$, and its degree in both $X$ and $Y$ is at most $\binom{\ell_t-1}{2}$. The linear factors of $P$ are the common linear factors of all the $R_t$'s, or equivalently the linear factors of $\gcd(R_1,\dotsc,R_s)$. 
One can thus apply standard algorithms to compute this gcd and then factor it in time polynomial in $\binom{k}{2}$. Moreover, $\mult{F}{P}=\mult{F}{\gcd(R_1,\dotsc,R_s)}$. Therefore, this describes an algorithm to compute the linear trinomial factors of $P$ and their multiplicities.
\end{proof}

\subsection{Finding multilinear factors} \label{sec:multilinear} 

\begin{theorem}\label{thm:multilinfact}
Let $\KK$ be a number field and
\[P=\sum_{j=1}^k a_j X^{\alpha_j} Y^{\beta_j}\in\KK[X,Y]\text.\]
There exists a deterministic polynomial time algorithm to compute all the multilinear factors of $P$, with their multiplicities.
\end{theorem}

\begin{proof}
As for computing linear factors, the task of computing multilinear factors can be split into three parts: computing the obvious monomial factors, computing the binomial factors using \Lemma{binomials}, and computing factors with at least three monomials. 

One can find all factors of the form $uXY+vX+wY+t$ with $wt\neq0$ using the Gap Theorem for multilinear factors (\Theorem{gap:multilin}). As for linear factors, one can compute a decomposition of $P$ with respect to both $X$ and $Y$ in the sense of \Definition{decomposition}: $P=X^{\alpha_{(1)}}Y^{\beta_{(1)}} R_1+\dotsb+X^{\alpha_{(s)}}Y^{\beta_{(s)}} R_s$ such that for all $F=uXY+vX+wY+t$  with $wt\neq0$, $\mult{F}{P}=\min(\mult{F}{R_1},\dotsc,\mult{F}{R_s})$. 
Moreover the sum of the degrees of the $R_t$'s (with respect to $X$ or $Y$) is bounded by $2\binom{k-1}{2}$. Then, finding such factors with their multiplicities is reduced to computing the multilinear factors of $\gcd(R_1,\dotsc,R_s)$ with multiplicities.

 There are two other cases to consider: $t=0$ and $w=0$.  
Both cases can be treated in the same way. Let us first concentrate on the case $t=0$. 

Suppose that $uXY+vX+wY$ divides $P$, $uvw\neq0$. Let $P_{XY}$ be 
the reciprocal polynomial of $P$ with respect to its two variables: 
\[P_{XY}=X^A Y^B P(1/X,1/Y)=\sum_{j=1}^k a_j X^{A-\alpha_j} Y^{B-\beta_j}\text,\]
where $A=\max_j \alpha_j$ and $B=\max_j\beta_j$. Then $uXY+vX+wY$ divides $P$ if and only if $u+vY+wX$ divides $P_{XY}$. Therefore, one can compute the linear factors of $P_{XY}$ with three monomials and deduce the factors of $P$ of the form $uXY+vX+wY$ from them. 

One treats in the same way the remaining case $w=0$. Let $P_X=X^AP(1/X,Y)$, we have that $uXY+vX+t$ divides $P$ if and only if $uY+v+tX$ divides $P_X$.
Therefore, this case is also reduced to the computation of linear factors.

Altogether, this gives a polynomial-time algorithm to compute multilinear factors of bivariate lacunary polynomials. 
\end{proof}

In the previous algorithm, factors with at least three monomials are computed in several steps, using different Gap Theorems. There are three cases: the general case $wt \neq 0$, and the special cases $t=0$ and $w=0$.
Each time, a first decomposition of the polynomial is computed with respect to the variable $X$, and then refined with respect to the variable $Y$. 

We aim to show that these three distinct steps can be performed in only one step. That is, one can first compute a decomposition with respect to the variable $X$ which is compatible with all the multilinear factors with at least three monomials, and then refine this decomposition with respect to the variable $Y$. 

The following proposition does not necessarily improve the running time of the algorithm and makes it more sequential but it simplifies its implementation. 
And most importantly, it will be crucial for the generalization to multivariate polynomials.

In the previous algorithms, some order on the monomials has been implicitly used. To compute a decomposition of a polynomial with respect to $X$, the monomials have to be ordered with an order compatible with the natural order on the exponents of the variable $X$. 
In particular, we assume in the next proposition that the indices are ordered such that $j<\ell$ implies $\alpha_j\le\alpha_\ell$ for all $j$ and $\ell$.

\begin{proposition}\label{prop:decomposition}
Let $P$ be as in \Theorem{multilinfact}. There exists a deterministic polynomial time algorithm to compute a decomposition of $P$ with respect to $X$, compatible with the set of all multilinear polynomials with at least three monomials, and of degree at most $3\binom{k-1}{2}$.
\end{proposition}

\begin{proof}
A decomposition $\Qcal=\{Q_1,\dotsc,Q_s\}$ is completely determined by a set of indices: for each $Q_t$, we consider its smallest index $j_t$. 
Then $J=\{j_1,\dotsc,j_s\}\subseteq\{1,\dotsc,k\}$ determines $\Qcal$. In the following, the decompositions are represented by their corresponding set of indices. To build a unique decomposition, we take the intersection of several decompositions viewed as sets of indices. If $J$ defines a decomposition compatible with a set $\Fcal$ then by Remark~\ref{rk:subdecomposition} any subset $I\subset J$ of indices defines a decomposition of $P$ compatible with $\Fcal$. Therefore, the intersection of two decompositions, respectively compatible with sets $\Fcal$ and $\mathcal G$, is compatible with $\Fcal\cup\mathcal G$.

It remains to show that the intersection of the decompositions compatible with the different kinds of multilinear factors with at least three monomials has degree at most $3\binom{k-1}{2}$. In the following, we consider only decompositions with respect to $X$. 
The decomposition $J=\{j_1,\dotsc,j_s\}$ compatible with $uXY+vX+wY+t$, $wt\neq 0$, is defined by $j_1=1$ and for all $t\ge 2$, $j_t$ is the smallest index such that $\alpha_{j_t}-\alpha_{j_{t-1}}>2\binom{j_t-j_{t-1}}{2}$.

The decompositions compatible with $uXY+vX+wY$ and $uXY+wY+t$ are the same, that is defined by a same set of indices. Indeed, a decomposition of $P$ compatible with $uXY+vX+wY$ is obtained using the Gap Theorem for linear factors on $P_{XY}$, while a decomposition of $P$ compatible with $uXY+wY+t$ is obtained using the same Gap Theorem on $P_X$. In both cases, the first gap is given by the largest $\ell$ such that $A-\alpha_\ell>A-\alpha_k+\binom{k-\ell}{2}$, that is $\alpha_k-\alpha_\ell>\binom{k-\ell}{2}$. Since both decompositions are obtained using the same equation, they are actually equal.

It remains to prove that the decomposition induced by $J\cap L$ has degree in $X$ at most $3\binom{k-1}{2}$. To this end, we assume that $J\cap L=\emptyset$ and 
prove that $\alpha_k-\alpha_1\le3\binom{k-1}{2}$. The result then follows by super-additivity of the function $k\mapsto 3\binom{k-1}{2}$. 
Since $J\cap L=\emptyset$, each gap of $L$ falls strictly between two gaps of $J$. This implies that the degree of $J\cap L$ is bounded by the sum of the degrees of $J$ and $L$, that is $3\binom{k-1}{2}$.
\end{proof}

\subsection{Generalization to multivariate polynomials}\label{sec:nvar}

We state the generalization to multivariate polynomials only for multilinear factors. This covers in particular the case of linear factors. 

Finding multilinear factors of multivariate polynomials can be performed in three distinct steps as in the case of bivariate polynomials. The first step is the obvious computation of the monomial factors. The second step deals with binomial factors and reduces to univariate factorization. The third step reduces the computation of multilinear factors with at least three monomials to low-degree factorization.

Let us begin with the third step, which is very close to the bivariate case. 

\begin{theorem}\label{thm:multivarfact}
Let $\KK$ be an algebraic number field and
\[P=\sum_{j=1}^k a_j X_0^{\alpha_{0,j}}\dotsm X_n^{\alpha_{n,j}}\in\KK[X_0,\dotsc,X_n]\]
There exists a deterministic 
algorithm that finds all the multilinear factors of $P$ with at least three monomials, together with their multiplicities. 
Its complexity is polynomial in $k^n$. 

A randomized variant of the algorithm achieves a complexity polynomial in $k$ and $n$. Furthermore, there exists a deterministic algorithm of same complexity for computing the linear factors of $P$ only.
\end{theorem}

\begin{proof}
The idea of the algorithm is as before to compute a decomposition of $P$, compatible with the set of multilinear factors with at least three monomials. To this end, we first compute a decomposition with respect to the variable $X_0$, and then refine the decomposition using the other variables, sequentially. The computation of the decomposition for each variable is based on \Proposition{decomposition}. 

Without loss of generality, we describe the computation of the decomposition with respect to the variable $X_0$. Any irreducible multilinear polynomial $F$ with at least three monomials can be written as $F=F_0X_0+F_1$ with $F_0,F_1\in\KK[X_1,\dotsc,X_n]$ and $F_1\neq0$. First assume that $F_0\neq0$. Since $F$ has at least three monomials, at least one of $F_0$ and $F_1$ must have two monomials. If $F_0$ has at least two monomials, there exists a variable, say $X_1$, such that $F_0=uX_1+v$ for some nonzero $uv\in\KK[X_2,\dotsc,X_n]$. If $F_1$ has two monomials, $F_1=wX_1+t$ for some nonzero $w,t\in\KK[X_2,\dotsc,X_n]$. In both cases, $F$ can be viewed as a polynomial in $X_0$ and $X_1$, with at least three monomials, with coefficients in $\KK[X_2,\dotsc,X_n]$. Therefore, one can view $P$ as a bivariate polynomial in $X_0$ and $X_1$ over the field $\KK(X_2,\dotsc,X_n)$ and apply \Proposition{decomposition} to compute a decomposition of $P$ with respect to $X_0$, compatible with all multilinear polynomials with at least three monomials. In particular, this decomposition is compatible with the multilinear polynomials in $X_0$, \dots, $X_n$ over $\KK$ with at least three monomials, and such that the coefficient of $X_0$ is nonzero. Now the decomposition is also compatible with factors $F\in\KK[X_1,\dotsc,X_n]$ since such factors have to divide every coefficient of $P$ viewed as an element of $\KK[X_1,\dotsc,X_n][X_0]$.

Applying this algorithm sequentially with respect to every variable gives a decomposition of $P$ compatible with all multilinear polynomials with at least three monomials. Its total degree is at most $3\binom{k-1}{2}$ in each variable and the total number of polynomials in the decomposition is bounded by the number $k$ of terms in $P$. One can therefore compute the gcd of the polynomials in the decomposition and the irreducible factorization of the gcd using classical algorithms. Then we return the multilinear factors, together with their multiplicities.
The main computational task of this algorithm is the final irreducible factorization. Using classical algorithms, one can perform this final step in time polynomial in $k^n$ deterministically and in time polynomial in $k$ and $n$ probabilistically~\citep{Kal85}. Actually, in the final step one only needs to compute the multilinear factors of the gcd that was computed. Though no deterministic polynomial-time algorithm is known for this computation, a deterministic polynomial-time algorithm is known in the case of linear factors~\citep{Vol15}.
\end{proof}

For the computation of binomial multilinear factors, we extend \Lemma{binomials}. The main difference comes from the fact that in \Lemma{binomials} we used the same algorithm four times, 
once for each possible choice of exponents. In the case of multivariate polynomials, there is an exponential number of choices of exponents. Thus the same strategy yields an algorithm of exponential complexity in the number of variables. Therefore, one has to determine in advance a smaller number of possible vectors of exponents. The proof is inspired by the proof of~\citet[Lemma~5]{KaKoi06}. 

\begin{theorem}
Let $\KK$ be an algebraic number field and
\[P=\sum_{j=1}^k a_j X_0^{\alpha_{0,j}}\dotsm X_n^{\alpha_{n,j}}\in\KK[X_0,\dotsc,X_n]\]
There exists a deterministic polynomial-time algorithm that finds all the multilinear factors of $P$ with two monomials, together with their multiplicities.
\end{theorem}

\begin{proof}
In this proof, we denote by $X$ the tuple of variables $(X_0,\dotsc,X_n)$.
The strategy is to 
compute a set of candidate pairs of monomials $(X^\beta,X^\gamma)$ such that $P$ may have a factor of the form $uX^\beta+vX^\gamma$, and then to actually compute the factors. For the first step, we write what it means for $uX^\beta+vX^\gamma$ to be a factor of $P$ and deduce conditions on $\beta$ and $\gamma$. The second step is then a reduction to finding linear factors of univariate polynomials. 

We begin with the first step. Let $F$ be an irreducible multilinear binomial. One can write 
\[F=u\prod_{i=0}^n X_i^{\beta_i} + v\prod_{i=0}^n X_i^{\gamma_i}\]
with $\beta_i,\gamma_i\in\{0,1\}$ and $\beta_i+\gamma_i\le 1$ for all $i$. Without loss of generality, let us assume that $\beta\neq0$, and let $i_0$ be an index such that $\beta_{i_0}=1$. 
Then $F$ is a factor of $P$ if and only if the polynomial
\[P(X_0,\dotsc,X_{i_0-1},-\frac{v}{u}\prod_{i\neq i_0} X_i^{\gamma_i-\beta_i},X_{i_0+1},\dotsc,X_n)\times \prod_{i\neq i_0} X_i^{A_{i_0}}\] 
vanishes, where $A_{i_0}=\max_j \alpha_{i_0,j}$. That is, $F$ divides $P$ if and only if
\[\sum_{j=1}^k a_j\left(-\frac{v}{u}\right)^{\alpha_{i_0,j}} \prod_{i\neq i_0} X_i^{\alpha_{i,j}+\alpha_{i_0,j}(\gamma_i-\beta_i)+A_{i_0}} =0\text.\]
In particular, the term for $j=1$ has to be canceled out by at least another term. In other words, there must exist $j\in\{2,\dotsc,k\}$ such that
\begin{equation}\label{eq:betagamma}
\forall i, \alpha_{i,1}+\alpha_{i_0,1}(\gamma_i-\beta_i)=\alpha_{i,j}+\alpha_{i_0,j}(\gamma_i-\beta_i)\text.
\end{equation}
Furthermore $\alpha_{i_0,1}\neq\alpha_{i_0,j}$, for it would imply that $\alpha_{i,1}=\alpha_{i,j}$ for all $i$ otherwise. Thus \Equation{betagamma} uniquely determines $\gamma_i-\beta_i$, hence uniquely determines both $\beta_i$ and $\gamma_i$. 

The variable $X_{i_0}$ does not play a special role in the previous discussion and the same reasoning applies with any variable $X_i$ such that $\beta_i\neq0$. In the same way, $\beta$ and $\gamma$ play symmetric roles. 
This means that for every $j\ge 2$, if the system defined by \Equation{betagamma} has a nonzero solution $(\beta,\gamma)$, it defines a candidate pair $(X^\beta,X^\gamma)$ and its symmetric pair $(X^\gamma,X^\beta)$. The symmetric pair is redundant and we may only consider the pair $(X^\beta,X^\gamma)$.
More precisely, for $j\ge 2$, if there exists an integer $q>0$ such that for all $i$, either $\alpha_{i,j}-\alpha_{i,1}=0$ or $\alpha_{i,j}-\alpha_{i,1}=\pm q$, then one can define $\beta_i$ and $\gamma_i$ for all $i$ as follows:
\[(\beta_i,\gamma_i)=\begin{cases}
    (0,0) & \text{if $\alpha_{i,j}-\alpha_{i,1}=0$,}\\
    (1,0) & \text{if $\alpha_{i,j}-\alpha_{i,1}=q$,}\\
    (0,1) & \text{if $\alpha_{i,j}-\alpha_{i,1}=-q$.}
\end{cases}\]
Therefore, we can compute at most $(k-1)$ candidate pairs. It remains to prove that, given a candidate pair $(X^\beta,X^\gamma)$, one can indeed compute all the factors of $P$ of the form $uX^\beta+vX^\gamma$. 

The algorithm for the second step is actually almost the same as in the proof of \Lemma{binomials}. 
Suppose that $\beta\neq0$. (Otherwise, invert $\beta$ and $\gamma$.) Let $i_0$ be any index such that $\beta_{i_0}=1$. Let 
\[Q(Y,X)=\prod_{i\neq i_0} X_i^{A_{i_0}} P(X_0,\dotsc,X_{i_0-1},Y\prod_{i\neq i_0} X_i^{\gamma_i-\beta_i},X_{i_0+1},\dotsc,X_n)\] 
where $A_{i_0}=\max_j\alpha_{i_0,j}$, viewed as an element of $\KK[Y][X_0,\dotsc,X_n]$. That is, let us write $Q=\sum_\delta q_\delta(Y) X^\delta$. 

Let $F=uX^\beta+vX^\gamma$ 
be a candidate factor. Then $F$ divides $P$ if and only if $G=uY+v$ divides $Q$, if and only if $G$ divides each $q_\delta$. More precisely, $\mult{F}{P}=\mult{G}{Q}=\min_\delta\mult{G}{q_\delta}$. Therefore, factors of the form $uX^\beta+vX^\gamma$ can be computed in deterministic polynomial-time by computing the factors of the univariate polynomials $q_\delta$, using Lenstra's algorithm~\citep{Len99}. 
\end{proof}

Note that in the previous algorithm, one could actually decrease the number of candidate pairs. To compute these pairs, we used the fact that the term for $j=1$ has to be canceled by at least another term. One could then use the same argument for every term: Each term has to be canceled by another one. Applying the same reasoning for all terms would imply more conditions on $(\beta,\gamma)$, thus potentially decrease the number of candidate pairs.

\subsection{Factorization in other fields of characteristic zero}\label{sec:otherFields} 

We have given algorithms to compute the multilinear factors of multivariate polynomials over number fields. Nevertheless, the Gap Theorems these algorithms are based on are valid over any field of characteristic zero. 
We can thus state our results in a more general form. Our algorithms are actually \emph{reductions} from the multivariate lacunary factorization problem to the univariate lacunary case on the one hand, and the multivariate low-degree case on the other hand. Such kind of reductions is fairly classical in polynomial factorization, see for instance \citep{Kal85}.

The proof of the following proposition is the same as the proof of \Theorem{multivarfact}. 

\begin{proposition}\label{prop:otherfields}
Let $\KK$ be a field of characteristic $0$, and
\[P = \sum_{j=1}^k a_j X_0^{\alpha_{0,j}}\dotsb X_n^{\alpha_{n,j}}\in\KK[X_0,\dotsc,X_n].\]
One can compute in deterministic polynomial time a polynomial $P_\text{low}$ of degree at most $3\binom{k-1}{2}$ such that the sets of multilinear factors (with multiplicities) having at least three monomials of $P$ and $P_\text{low}$ are equal. 
\end{proposition}

Using this proposition, one can obtain polynomial-time algorithms to compute the multilinear factors with at least three monomials of lacunary polynomials over several fields of characteristic $0$. The following corollary gives some examples of fields for which one obtains such algorithms.

\begin{corollary}
There exist randomized 
algorithms which given some polynomial $P\in\QQ[X_1,\dots,X_n]$, compute its multilinear factors with at least three monomials with coefficients in $\Qbar$, $\CC$ and $\QQp$, respectively. These algorithms run in time polynomial in the size of lacunary representation of $P$.
\end{corollary}

\begin{proof}
In all cases, the corollary follows from \Proposition{otherfields} using a factorization algorithm that is polynomial in the degree of the input polynomial. For references to the numerous algorithms for the factorization over $\Qbar$ or $\CC$, we refer the reader to \citep{CheGa05}. An algorithm for the factorization over $\QQp$ is due to~\citet{Chi94}. 
\end{proof}

To compute binomial factors, we need a univariate lacunary factorization algorithm. To the best of our knowledge, algorithms for this task are only known over algebraic number fields. Note that over algebraic closed fields, there is no hope of computing these factors in polynomial time: The number of univariate factors equals the degree of the lacunary polynomial and is thus exponential in the size of the input.

\subsection{Comparison with previous work} \label{sec:comparison}

We aim to compare our techniques with the ones of \citet{KaKoi05,KaKoi06}. 
Their first result deals with linear factors of bivariate lacunary polynomials over the rational numbers, while the second ones is an extension to the case of low-degree factors, with multiplicities, of multivariate lacunary polynomials over number fields. Thus, our algorithm applies to more fields: nonzero characteristic (see Section~\ref{sec:posChar}) or absolute factorization for instance. On the other hand, it is limited to multilinear factors rather than arbitrary low-degree factors.

Let us now compare the techniques used and the complexities of both algorithms. 
For the sake of simplicity, we focus on the following case: computing linear factors of bivariate polynomials over the integers or the rationals. We first briefly highlight the main differences and then perform a somewhat detailed complexity analysis of both algorithms. 

Binomial factors are computed essentially in the same way in both algorithms, using Lenstra's algorithm for univariate lacunary polynomials~\citep{Len99}.
Let us then focus on the task of finding trinomial factors. 
The high-level structure of both algorithms is the same. Both of them identify gaps in the exponent vector of the input polynomial, and use these gaps to write the polynomial as a sum of low-degree polynomials. The linear factors of the input polynomial are then the common linear factors of all these low-degree polynomials. Two differences make our algorithm simpler: First our computation of the gaps is based on purely combinatorial considerations on the exponents while one needs to compute the height of the polynomial in Kaltofen and Koiran's algorithm; Secondly our algorithm directly gives the multiplicities of the factors while their algorithm has to be applied to the successive (sparse) derivatives of the input polynomial in order to get the multiplicities.

To analyse the complexities of both algorithms, we remark that the main computational task is in both cases the final step when the common linear factors of the low-degree polynomials must be computed. For this step, one can either compute the gcd of the low-degree polynomials and its linear factors, or compute the linear factors of one of the low-degree polynomials and then use a divisibility test on the other polynomials. Whichever strategy one chooses, a relevant measure to analyse the complexity of any of the two algorithms is the size of the low-degree polynomials. Their coefficients are coefficients of the input polynomial. Their degree is related to the Gap Theorems. In our case, Theorem~\ref{thm:gap} directly shows that the low-degree polynomials have degree at most $O(k^2)$ where $k$ is the number of terms of the input polynomial. On the other hand, \citet{KaKoi05} prove a gap of size $O(\log(k)+h)$ between consecutive powers, where $h$ is the (logarithmic) height of the input polynomial, which translates into a bound $O(k(\log k+h))$ on the degree of the low-degree polynomials to be factored. In the case of integer coefficients, the height is bounded by the bitsize $b$ of the coefficients. To compute the height in the case of rational coefficients, one first needs to reduce all coefficients to the same denominator and the height is bounded by $O(kb)$. In other words, the bound on the degree of the low-degree polynomials to be factored in Kaltofen and Koiran's algorithm is $O(kb)$ (ignoring logarithmic factors) for integer coefficients and $O(k^2b)$ for rational coefficients.

The best general factorization algorithm for degree-$d$ integer bivariate polynomials with coefficients of bitsize $b$ has complexity $O(d^6b^5)$~\citep{Lec10,HaHoeNo11}. To compute the sole linear factors, the best strategy has the much lower complexity $O(d^2b)$~(see \citep{AECF} for a complete complexity analysis). For our algorithm, this means that the complexity of the factorization step is bounded by $O(k^4b)$ in the case of integer coefficients. In Kaltofen and Koiran's algorithm, the same step runs in time $O(k^2b^3)$ (ignoring logarithmic factors in $k$). Since computing the multiplicity of the factors in their approach requires to apply the same algorithm to at most $(k-1)$ derivatives of the polynomial, the total complexity is $O(k^3b^3)$. In other word, our algorithm is faster as soon as $b\ge\sqrt k$ in the case of integer coefficients. To compute the linear factors of rational polynomials, the coefficients have to be reduced to the same denominator and the coefficient size can increase by a factor $d$. That is, computing the linear factors of rational polynomials has complexity $O(d^3b)$. Using the bound of the previous paragraph and the same reasoning as above, the complexity of our algorithm becomes $O(k^6b)$ and the complexity of Kaltofen and Koiran's algorithm becomes $O(k^7b^4)$. Thus the worst-case complexity of our algorithm is better for rational coefficients.

\section{Positive characteristic}\label{sec:posChar} 

To extend the previous results to positive characteristic, one needs an equivalent of \Theorem{val}. Unfortunately, \Theorem{val} does not hold in positive characteristic. In characteristic $2$, the polynomial $(1+X)^{2^n}+(1+X)^{2^{n+1}}=X^{2^n}(X+1)$ has only two terms, but its valuation equals $2^n$. Therefore, its valuation cannot be bounded by a function of the number of terms. Note that this can be generalized to any positive characteristic. In characteristic $p$, one can consider the polynomial $\sum_{i=1}^p (1+X)^{p^{n+i}}$. 

Nevertheless, the exponents used in all these examples depend on the characteristic. In particular, the characteristic is always smaller than the largest exponent that appears. We shall show that in large characteristic, \Theorem{val} still holds and can be used to give factorization algorithms. This contrasts with the previous results that use the notion of height of an algebraic number, hence are not valid in any positive characteristic. 

In fact, \Theorem{val} holds as soon as $\wronskian(f_1,\dots,f_k)$ does not vanish. 
The difficulty in positive characteristic is that it is not true anymore that the Wronskian does not vanish as soon as $(f_j)_j$ is a linearly independent family. Consider for instance the family $f_1=1$ and $f_2=X^2$ in characteristic $2$. 
Yet, the Wronskian is still related to linear independence by the following result, see \citep{Kaplansky}:
\begin{proposition}\label{prop:positive}
Let $\KK$ be a field of characteristic $p$ and $f_1,\dots,f_k\in\KK[X]$. Then $f_1$, \dots, $f_k$ are linearly independent \emph{over $\KK[X^p]$} if and only if their Wronskian does not vanish.
\end{proposition}

This allows us to give an equivalent of \Theorem{val} in large positive characteristic.

\begin{theorem}\label{thm:valPosChar}
Let $P=\sum_{j=1}^k a_j X^{\alpha_j}(1+X)^{\beta_j} \in\KK[X]$ 
with $\alpha_1\le\dotsb\le\alpha_k$. If 
the characteristic $p$ of $\KK$ satisfies
$p>\max_j(\alpha_j+\beta_j)$, then the valuation of $P$ is at most $\max_j (\alpha_j+\binom{k+1-j}{2})$, 
provided $P$ is nonzero. 
\end{theorem}

\begin{proof}
Let 
$f_j=X^{\alpha_j}(1+X)^{\beta_j}$ for $1\le j\le k$. 
The proof of \Theorem{val} has two steps: We prove that we can assume that the Wronskian of the $f_j$'s does not vanish, and under this assumption we get a bound of the valuation of the polynomial. The second part only uses the non-vanishing of the Wronskian and can be used here too. We are left with proving that the Wronskian of the $f_j$'s can be assumed to be nonzero when the characteristic is large enough.

Assume that the Wronskian of the  $f_j$'s 
is zero: 
By \Proposition{positive}, there is a vanishing linear combination of the $f_j$'s 
with coefficients $b_j$ in
$\KK[X^p]$. Let us write $b_j= \sum b_{i,j}X^{ip}$. Then  $\sum_i X^{ip}\sum_j b_{i,j}f_j=0$.
Since $\deg f_j=\alpha_j+\beta_j<p$, $\sum_j b_{i,j}f_j=0$ for all $i$.
We have thus proved that there is a linear combination of the $f_j$'s equal to zero with coefficients
in $\KK$. Therefore, we can assume we have a basis of the $f_j$'s whose Wronskian
is nonzero and use the same argument as for the  characteristic zero.
\end{proof}

Based on this result, the algorithms we develop in characteristic zero for PIT and factorization can be used for large enough characteristics. 
Computing with lacunary polynomials in positive characteristic has been shown to be hard in many cases~\citep{vzGaKaShp96,KaShp99,KiSha99,KaKoi05,BiCheRo12,KaLe13}. In particular, \citet{BiCheRo12} have recently shown that it is $\NP$-hard to find roots in $\FFp$ for polynomials over $\FFp$. 

Let $\FFps$ be the field with $p^s$ elements for $p$ a prime number and $s>0$. It is represented as $\FFp[\xi]/\langle\varphi\rangle$ where $\varphi$ is a monic irreducible polynomial of degree $s$ with coefficients in $\FFp$. As for number fields, $\varphi$ can be given as input of the algorithms, and a coefficient $c\in\FFps$ is represented by a polynomial of degree smaller than $\deg(\varphi)$.

\begin{theorem}
Let $\FFps$ be a finite field, and
\[P=\sum_{j=1}^k a_j X^{\alpha_j}(uX+v)^{\beta_j} \in\FFps[X]\text,\] 
where $p>\max_j(\alpha_j+\beta_j)$.
There exists a 
deterministic algorithm to test if $P$ vanishes identically, 
whose complexity is polynomial in the bitsizes of the $\alpha_j$'s and $\beta_j$'s, in $k$ and $s\log p$.
\end{theorem}

\begin{proof}[Proof idea]
The proof of this theorem is very similar to the proof of \Theorem{pit}, using \Theorem{valPosChar} instead of \Theorem{val}. 
The main difference 
occurs when $u=0$ or $v=0$. In these cases, we rely in characteristic zero on Lenstra's algorithm to test sums of the form $\sum_j a_jv^{\beta_j}$ for zero. There is no equivalent of Lenstra's algorithm in positive characteristic, but these tests are actually much simpler. These sums can be evaluated using repeated squaring in time polynomial in $\log(\beta_j)$, that is polynomial in the input length.

The basic operations in the algorithm are operations in the ground field $\FFp$. Therefore, the result also holds if bit operations are considered. The only place where computations in $\FFps$ have to be performed in the algorithm is the tests for zero of coefficients of the form $\sum_j\binom{\alpha_j}{\ell_j} a_j u^{-\alpha_j}(-v)^{\ell_j}$ where the $\alpha_j$'s and $\ell_j$'s are integers and $a_j\in\FFps$, and the sum has at most $k$ terms. The binomial coefficient is to be computed \emph{modulo} $p$ using for instance Lucas' Theorem~\citep{Lucas1878}.
\end{proof}

Note that the condition $p>\max_j(\alpha_j+\beta_j)$ means that $p$ has to be greater than the degree of $P$. This condition is a fairly natural condition for many algorithms dealing with polynomials over finite fields, especially prime fields, for instance for root finding algorithms~\citep{BiCheRo12}.

We now turn to the problem of factoring lacunary polynomials with coefficients in fields of large characteristic. We state it in the most general case of finding multilinear factors of multivariate polynomials.

\begin{theorem}\label{thm:FactorPosChar}
Let $\FFps$ be the field with $p^s$ elements, and 
\[P=\sum_{j=1}^k a_j X_0^{\alpha_{0,j}} \dotsb X_n^{\alpha_{n,j}}\in\FFps[X_0,\dotsc,X_n]\text,\]
where $p>\max_j(\alpha_j+\beta_j)$.  There exists a probabilistic polynomial-time algorithm to find all the multilinear factors of $P$ 
with at least three monomials, together with their multiplicities.

On the other hand, deciding whether $P$ has a binomial factor is $\NP$-hard under randomized reductions. More precisely, for every pair of relatively prime multilinear monomials $(X^\beta,X^\gamma)$, deciding whether there exist nonzero $u$ and $v$ such that $uX^\beta+vX^\gamma$ divides $P$ is $\NP$-hard under randomized reductions.
\end{theorem}

\begin{proof}
The second part of the theorem is the consequence of the $\NP$-hardness (under randomized reductions) of finding roots in $\FFps$ of lacunary univariate polynomials with coefficients in $\FFps$~\citep{KiSha99,BiCheRo12,KaLe13}: Let $Q$ be a lacunary univariate polynomial over $\FFps$, and let $d=\deg(Q)$. Let us define $P(X_0,\dotsc,X_n)=(X^\beta)^d Q(X^{\gamma-\beta})$ where $X^{\gamma-\beta}=\prod_i X_i^{\gamma_i-\beta_i}$. Then $P$ is a polynomial. We aim to show that $Q$ has a nonzero root if and only if $P$ has a binomial factor of the form $uX^\beta+vX^\gamma$. Let $F=uX^\beta+vX^\gamma$. Without loss of generality, we can assume that $\beta\neq0$ and $\beta_0=1$. Then $F$ divides $P$ if and only if 
\[(X^\gamma)^{\max_j\alpha_{0,j}} P\left(-\frac{v}{u}\prod_{i>0} X_i^{\gamma_i-\beta_i},X_1,\dotsc,X_n\right)=0\text.\]
Let $X^\delta=(X^\gamma)^{\max_j\alpha_{0,j}} (X^\beta)^d$. Since $\beta_0=1$ and $\gamma_0=0$, the previous equality is equivalent to
\[X^\delta Q\left(\left(-\frac{v}{u}\prod_{i>0} X_i^{\gamma_i-\beta_i}\right)^{-1} \prod_{i>0} X_i^{\gamma_i-\beta_i}\right)=X^\delta Q(-\frac{u}{v})=0\text.\]
In other words, this is equivalent with the fact that $-u/v$ is a root of $Q$. Deciding whether $uX^\beta+vX^\gamma$ divides $P$ is thus $\NP$-hard under randomized reductions.

For the first part, the algorithm we propose is actually the same as in characteristic zero (\Theorem{linfact}). This means that it relies on known results for factorization of dense polynomials. Yet, the only polynomial-time algorithms known for factorization in positive characteristic are probabilistic~\citep{vzGaGe13}. Therefore our algorithm is probabilistic and not deterministic as in characteristic zero. 
\end{proof}

\end{document}